\def\Title{Quantum Reality and Measurement: A Quantum Logical Approach}
\def\Author{Masanao Ozawa}
\def\Affiliation{Graduate School of Information Science, Nagoya University, 
Nagoya 464-8601, Japan}
\def\Abstract{
The recently established universal uncertainty principle revealed that two 
nowhere commuting observables can be measured simultaneously in some state,
whereas they have no joint probability distribution in any state.
Thus, one measuring apparatus can simultaneously measure two observables that 
have no simultaneous reality.
In order to reconcile this discrepancy, 
an approach based on quantum logic
is proposed to establish 
the relation between quantum reality and measurement.
We provide a language speaking of values of observables 
independent of measurement based on quantum logic and 
we construct in this language the state-dependent notions of joint determinateness,
value identity, and simultaneous measurability.
This naturally provides a contextual interpretation, in which 
we can safely claim such a statement 
that one measuring apparatus measures one observable in one context 
and simultaneously it measures another nowhere commuting observable 
in another incompatible context. 
}
\def\Keywords{quantum logic, 
quantum set theory,
quantum measurement,
joint determinateness, 
simultaneous measurability, 
contextual interpretation}
\def\PACSnumbers{03.65.Ta, 02.10.Ab , 03.65.Ud}
\documentclass[12pt]{article}
\evensidemargin=\oddsidemargin
\usepackage{amsthm}
\newcommand{\cut}[1]{}
\newcommand{\ignore}[1]{}
  \newcommand{\beq}{\begin{equation}}
  \newcommand{\eeq}{\end{equation}}
  \newcommand{\beql}[1]{\begin{equation}\label{eq:#1}}

  \newcommand{\beqa}{\begin{eqnarray}}
  \newcommand{\eeqa}{\end{eqnarray}}
  \newcommand{\beqas}{\begin{eqnarray*}}
  \newcommand{\eeqas}{\end{eqnarray*}}
  \newtheorem{theorem}{Theorem}
  \newcommand{\C}{{\bf C}}

  \newcommand{\R}{{\bf R}}

  \newcommand{\bM}{{\bf M}}

  \newcommand{\bP}{{\bf P}}
  
  \newcommand{\bS}{{\bf S}}

  \newcommand{\cC}{{\mathcal C}}

  \newcommand{\cH}{{\mathcal H}}
  \newcommand{\cI}{{\mathcal I}}
  
  \newcommand{\cK}{{\mathcal K}}
  \newcommand{\cL}{{\mathcal L}}

  \newcommand{\cO}{{\mathcal O}}
  \newcommand{\cP}{{\mathcal P}}
  \newcommand{\cQ}{{\mathcal Q}}
  \newcommand{\cR}{{\mathcal R}}
  \newcommand{\cS}{{\mathcal S}}

  \newcommand{\da}{\dagger}
  
  \newcommand{\de}{\delta}

  \newcommand{\mb}{\mbox}

  \newcommand{\ph}{\phi}
 \def\ps{\psi}
  \newcommand{\rh}{\rho}
  \newcommand{\si}{\sigma}
  
  \newcommand{\And}{\wedge}

  \newcommand{\De}{\Delta}

  \newcommand{\Eq}[1]{Eq.~(\ref{eq:#1})}

  \newcommand{\Not}{\neg}
  
  \newcommand{\Or}{\vee}

  \newcommand{\Sup}{\bigvee}

  \newcommand{\Tr}{\mbox{\rm Tr}}

  \newcommand{\beqan}{\begin{eqnarray*}}
  \newcommand{\beqar}[1]{\begin{equation}\label{#1}\begin{array}{l}}
  
  \newcommand{\bx}{{\bf x}}

  \newcommand{\eeqar}{\end{array}\end{equation}}
  \newcommand{\eq}[1]{(\ref{eq:#1})}

\newcommand{\ket}[1]{|#1\rangle}
\newcommand{\ketbra}[1]{|#1\rangle\langle#1|}

  \newcommand{\ran}{\mbox{\rm ran}}

\newcommand{\bmat}{\left[\begin{array}{rr}}
\newcommand{\emat}{\end{array}\right]}
\newcommand{\bvec}{\left[\begin{array}{r}}
\newcommand{\evec}{\end{array}\right]}

\newcommand{\val}[1]{[\![#1]\!]}

\newcommand{\cm}{\mathrm{com}}

\newcommand{\Sp}{{\rm Sp}}

\newcommand{\M}{\bM}
\renewcommand{\ran}{\cR}
\title{\bf \Title}
\author{\rm \Author \\ \\
\small\rm \Affiliation}
\date{}
\begin{document}
\maketitle
\begin{abstract}
\Abstract
PACS: \PACSnumbers
\end{abstract}

\section{Introduction}
\label{se:1}
In quantum mechanics, the value of any individual observable is postulated to be 
measured precisely.
The values of several commuting observables can be
reduced in principle by the functions of the value of a single observable, so
that they are simultaneously measurable, and their joint probability
distribution is predicted by the Born statistical formula.
However, 
there is a generic difficulty in considering the
values of non-commuting observables, or in considering their simultaneous
measurement.  

It has been accepted
that two observables are simultaneously measurable if and only if they
commute. However, this formulation is based on the state-independent
formulation in the sense that two observables are simultaneously measurable
{\em in any state} if and only if they commute {\em on the whole state space}. 
The recently established universal uncertainty principle 
\cite{03UVR,03HUR,03UPQ,04URN} revealed that two 
nowhere commuting observables can be measured simultaneously in some state 
\cite{07SMN},
whereas they have no joint probability distribution in any state. 
Thus, a single measuring apparatus can simultaneously measure two observables that 
have no simultaneous reality.  
In order to reconcile this discrepancy, we consider a fundamental question as to
the relation between quantum reality and measurement.

Here, an approach is proposed to establish the notion of the value of an observable 
independent of the measurement but obtained by the measurement.  
Our approach is based on quantum logic or quantum set theory 
\cite{Ta81,07TPQ} to deal with quantum reality. 
First, we provide a language for the values of observables based on quantum logic and 
we construct in this language a sentence that means that two observables have the same 
value in a given state \cite{05PCN,06QPC}.  
Then, a contextual interpretation can be naturally obtained and within that we can safely 
claim that one measuring apparatus can reproduce the value of an observable $A$ before
measurement as the meter value after measurement in one context, and simultaneously that
the same apparatus can reproduce the value of another nowhere commuting observable $B$ 
in another context.
This actually happens in the EPR state and in this interpretation
the simultaneous measurability of nowhere commuting observables 
does not conclude the simultaneous reality of the values of 
those observables, 
for which once a hidden variable theory was called for.

In Section \ref{se:2} we introduce the observational propositions, a propositional language 
constructed from atomic sentences meaning that a certain observable has a certain value,
and assign them the truth values and probabilities based on standard rules in quantum logic.
In Section \ref{se:3} we define a new observational proposition in the above language
meaning that given observables are jointly determinate in a given state,
and show some basic properties of this proposition.
In Section \ref{se:4} we define a new observational proposition in the above language
meaning that given two observables have the same value in a given state,
and show some basic properties of this proposition.
In Section \ref{se:5} we introduce a mathematical theory of quantum measurements
based on the notions of measuring processes, instruments, and POVMs.
In Section \ref{se:6} we define the notion of a measurement of an observable in a given
state based on the observational proposition introduced in Section \ref{se:4} so that
a measuring process is a measurement of an observable in a given state if and only if
the measured observable before the measurement and the meter observable 
after the measurement
has the same value in the given state,
and show some basic properties of this proposition.
In Section \ref{se:7} we define the notion of a simultaneous measurability of given 
observables in a given state based on the notion of a measurement of an observable in a given
state introduced in Section \ref{se:6}.  
We clarify the conceptual difference between simultaneous measurability 
and joint determinateness by the existence of $\R^2$-valued POVMs with certain properties.
We also show some cases in which nowhere commuting observables are simultaneously
measurable.
Section 8 is devoted to discussions and conclusions.

\section{Logic of observational propositions}
\label{se:2}

In this  paper, we assume for simplicity 
that every Hilbert space $\cH$ is finite dimensional and describes
a quantum system $\bS(\cH)$ unless stated otherwise.
The infinite dimensional case will be discussed elsewhere.
The observables are defined as self-adjoint operators, the states are
defined as density operators, and a vector state $\psi$ is identified
with the state $\ketbra{\psi}$.
We denote by $\cO(\cH)$ the set of observables,
by $\cS(\cH)$ the space of density operators, and 
by $\cL(\cH)$ the space of linear operators on $\cH$.
We  denote by $\R$ the set of real numbers.

We define {\em observational propositions} for $\cH$ 
by the following rules.

(R1) For any $X\in\cO(\cH)$ and $x\in \R$, the expression
$X=x$ is an observational proposition, called an {\em atomic 
observational proposition}.

(R2) If $\ph_1$ and $\ph_2$ are observational propositions,
$\Not \ph_1$ and $\ph_1\And \ph_2$ are also observational propositions.

Thus, every observational proposition is built up from atomic 
observational propositions by means of the connectives 
$\Not$ and $\And$.
We introduce the connective $\Or$ by definition.

(D1) $\ph_1\Or\ph_2:= \Not(\Not \ph_1\And\Not\ph_2)$.

We will freely use parentheses to clarify the construction.
For example, if $X_1,X_2,X_3\in\cO(\cH)$ and $x_1,x_2,x_2\in\R$,
then $(X_1=x_1\And X_2=x_2)\Or X_3=x_3$ is an
observational proposition. 

 The set of linear subspaces of a Hilbert space $\cH$ 
is a complete complemented modular lattice with the orthogonal complementation
$M\mapsto M^{\perp}$ \cite{BN36}.
The lattice operations satisfy
$M_1\And M_2=M_1\cap M_2$ and
$M_1\Or M_2=M_1+M_2$.
An operator $P$ is called a {\em projection} if $P=P^{\dagger}=P^2$.
The projection with range $M$ is denoted by $\cP (M)$
and the range of an operator $X$ is denoted by $\cR(X)$.
Then, we have $\cP(\cR(P))=P$ and $\cR(\cP(M))=M$ for all
projections $P$ and subspaces $M$, so that
the projection operators and the linear subspaces are in 
one-to-one correspondence, 
and the lattice structure is naturally introduced in the projections.
We call the lattice of projections on $\cH$ as the {\em quantum logic} of $\cH$
and denote it by $\cQ(\cH)$.
For any $X\in\cO(\cH)$ and $x\in\R$, we define the {\em spectral projection}
$E^{X}(x)$ by $E^{X}(x)=\cP(\{\ps\in\cH\mid X\ps=x\ps\})$.

For each observational proposition $\ph$, we assign its truth
value $\val{\ph}\in\cQ(\cH)$ by the following rules.

(T1) $\val{X=x}=E^{X}(x).$

(T2) $\val{\Not \ph}=\val{\ph}^{\perp}.$

(T3) $\val{\ph_1\And \ph_2}=\val{\ph_1}\And\val{\ph_2}.$

From (D1),  (T2) and (T3), we have

(T4) $\val{\ph_1\Or \ph_2}=\val{\ph_1}\Or\val{\ph_2}.$

We have $\val{(\ph_1\And\ph_2)\And \ph_3}=\val{\ph_1\And(\ph_2\And \ph_3)}$,
so that we do not distinguish the observational propositions $(\ph_1\And\ph_2)\And \ph_3$
and $\ph_1\And(\ph_2\And \ph_3)$ to denote them by $\ph_1\And\ph_2\And \ph_3$.
Similar conventions are also applied to longer propositions and the connective $\Or$.

We define the {\em probability} $\Pr\{\ph\|\rh\}$ 
of an observational proposition $\ph$ in a state $\rh$ by
$\Pr\{\ph\|\rh\}=\Tr[\val{\ph}\rh]$.
We say that {\em an observational proposition $\ph$ holds in a state $\rh$} if
$\Pr\{\ph\|\rh\}=1$.

Suppose that $X_1,\ldots,X_n\in\cO(\cH)$ are mutually commuting.
Let $x_1,\ldots,x_n\in\R$.
We have
\beqas
\val{ X_1=x_1\And\cdots\And X_n=x_n}
&=&\val{ X_1=x_1}\And\cdots\And\val{X_n=x_n}\\
&=&E^{X_1}(x_1)\And\cdots\And E^{X_n}(x_n)\\
&=&E^{X_1}(x_1)\cdots E^{X_n}(x_n).
\eeqas
Hence, we reproduce the {\em Born statistical formula}
\beqa
\Pr\{X_1=x_1\And\cdots\And X_n=x_n\|\rh\}
&=&\Tr[E^{X_1}(x_1)\cdots E^{X_n}(x_n)\rh].
\eeqa
For any polynomial $p(X_1,\ldots,X_n)$  
we also have 
\beqa
\lefteqn{
\Tr[p(X_1,\ldots,X_n)\rh]}\quad\nonumber\\
&=&
\sum_{(x_1,\ldots,x_n)\in\R^{n}}p(x_1,\ldots,x_n)
\Pr\{X_1=x_1\And\cdots\And X_n=x_n\|\rh\}
\eeqa
for any state $\rh$.
From the above, our definition of the truth vales of observational propositions 
are consistent with the standard quantum mechanics.

\section{Joint determinateness}
\label{se:3}

For observational propositions $\ph_1,\ldots,\ph_n$, we define the
observational proposition $\Sup_{j}\ph_j$ by
$\Sup_{j}\ph_j=\ph_1\Or\cdots\Or\ph_n$.
We denote by $\Sp(X)$ the set of eigenvalues of  an observable $X$.

\sloppy
For any observables
$X_1,\ldots,X_n$ we define the observational proposition
$\cm(X_1,\ldots,X_n)$ by
\beqa\label{eq:SD}
\cm(X_1,\ldots,X_n)
:=
\Sup_{x_1\in\Sp(X_1),\ldots,x_n\in\Sp(X_n)} 
X_1=x_1\And\cdots\And X_n=x_n.
\eeqa
We say that observables $X_1,\ldots,X_n$ are {\em jointly determinate} 
in a state $\rh$ if the observational proposition
$\cm(X_1,\ldots,X_n)$  holds in $\rh$.
In general, we say that observables $X_1,\ldots,X_n$ are {\em jointly determinate} 
in a state $\rh$ with probability $\Pr\{\cm(X_1,\ldots,X_n)\|\rh\}$.

Then, we have the following.
\begin{theorem}\label{th:SD}
Observables  $X_1,\ldots,X_n$ are jointly determinate in
a vector state $\ps$  if and only if the state $\ps$ is a superposition
of common eigenvectors of $X_1,\ldots,X_n$.
\end{theorem}
\begin{proof}
It is easy to see that $\ps\in\cR(\val{X_1=x_1\And\cdots\And X_n=x_n})$
if and only if $\ps$ is a common eigenstate belonging to eigenvalues $x_1,\ldots,x_n$
for observables $X_1,\ldots,X_n$, respectively, so that the assertion follows easily.
\end{proof}

Two observables $X$ and $Y$ are said to be {\em nowhere commuting} if there is no
common eigenstate.  From Theorem \ref{th:SD}, $X$ and $Y$ are nowhere commuting
if and only if $\val{\cm(X,Y)}=0$.  

\sloppy
A probability distribution $\mu(x_1,\ldots,x_n)$ on $\R^{n}$,
i.e., $\mu:\R^{n}\to [0,1]$ and  
$\sum_{(x_1,\ldots,x_n)\in\R^{n}}\mu(x_1,\ldots,x_n)=1$,
is called a {\em joint probability distribution} of 
$X_1,\ldots,X_n\in\cO(\cH)$ in $\rh\in\cS(\cH)$ if
\beqa
\mu(x_1,\ldots,x_n)
=\Pr\{X_1=x_1\And\cdots\And X_n=x_n\|\rh\}.
\eeqa
It is easy to see that a joint probability distribution $\mu$
of $X_1,\ldots,X_n$ in $\rh$ is unique, if any.

The notion of joint probability distributions is inherently a state-dependent
notion. It is well-known that the existence of the joint probability distribution
in any state is equivalent to the commutativity of observables under consideration
\cite{Gud68}.
Since the joint determinateness is naturally considered to be the state-dependent 
notion of commutativity, it is naturally expected that the joint determinateness
is equivalent to the state-dependent existence of the joint probability distribution,
as shown below.
 
\sloppy
\begin{theorem}\label{th:JPD}
Observables $X_1,\ldots,X_n$ are
jointly determinate in a state $\rh$
if and only if
there exists a joint probability distribution of $X_1,\ldots,X_n$ in
$\rh$. 
In this case, for any polynomial $p(X_1,\ldots,X_n)$ of observables
$X_1,\ldots,X_n$, we have
\beqa
\lefteqn{
\Tr[p(X_1,\ldots,X_n)\rh]}\quad\nonumber\\
&=&
\sum_{(x_1,\ldots,x_n)\in\R^{n}}p(x_1,\ldots,x_n)
\Pr\{X_1=x_1\And\cdots\And X_n=x_n\|\rh\}.
\eeqa
\end{theorem}
\begin{proof}
The assertion follows from an argument similar to 
the proof of Theorem 3.3 of Gudder \cite{Gud68}.
\end{proof}

Let $\ph(X_1,\ldots,X_n)$ be an observational proposition 
that includes symbols for and only for observables $X_1,\ldots,X_n$.
Then, $\ph(X_1,\ldots,X_n)$ is said to be {\em contextually well-formed} in a state $\rh$
if $X_1,\ldots,X_n$ are jointly determinate in $\rh$.
The following is an easy consequence from the transfer 
principle in quantum set theory \cite{07TPQ}.

\begin{theorem}
Let $\ph(X_1,\ldots,X_n)$ be an observational proposition 
that includes symbols for and only for observables $X_1,\ldots,X_n$.
Suppose that $\ph(X_1,\ldots,X_n)$ is a tautology in classical logic.
Then, we have 
\beqa
\val{\cm(X_1,\ldots,X_n)}\le\val{\ph(X_1,\ldots,X_n)}.
\eeqa
In particular, if $\ph(X_1,\ldots,X_n)$ is contextually well-formed
in a state $\rh$, then $\ph(X_1,\ldots,X_n)$ holds in $\rh$.
\end{theorem}

\section{Value identity of observables}
\label{se:4}

For any observables $X,Y$, we define 
the observational proposition $X=Y$ by 
\beqa\label{eq:ID}
X=Y:=
\Sup_{x\in \Sp(X)}X=x\And Y=x.
\eeqa
We say that  
observables $X$ and $Y$ {\em have the same value in a state $\rh$}
if $X=Y$ holds in $\rh$.
In this case, we shall write $X=_\rh Y$.
In general, we say that 
observables $X$ and $Y$ {\em have the same value in a state $\rh$
with probability $\Pr\{X=Y\|\rh\}$}.

\begin{theorem}\label{th:ID}
For any $X,Y\in\cO(\cH)$ and $\rh\in\cS(\cH)$, we have 
$X=_\rh Y$ if and only if there exists a joint probability distribution $\mu$
of $X,Y$ in $\rh$ such that 
\beqa\label{eq:ID-1}
\sum_{x\in\R}\mu(x,x)=1.
\eeqa
\end{theorem}
\begin{proof}
Suppose $X=_\rh Y$. We have $\Pr\{\cm(X,Y)\|\rh\}=1$, so that from
Theorem \ref{th:JPD}, $X$ and $Y$ have the joint probability distribution
$\mu$,  which satisfies \Eq{ID-1}.
Conversely, if we have the joint probability distribution $\mu$ satisfying
\Eq{ID-1}, we have
\beqas
\Pr\{X=Y\|\rh\}
&=&\Tr\left[\left(\Sup_{x\in \Sp(X)}E^{X}(x)\And E^{Y}(x)\right)\rh\right]\\
&=&\sum_{x\in\R}\Tr[(E^{X}(x)\And E^{Y}(x))\rh]
=\sum_{x\in\R}\mu(x,x)=1.
\eeqas
Thus, we have $X=_\rh Y$.
\end{proof}

In order to consider the state-dependent notion of measurement of observables, 
the notion of quantum perfect correlation, or quantum identical 
correlation, between two observables has been previously introduced in Ref.~\cite{05PCN,06QPC}  
as follows. We say that two observables $X,Y\in\cO(\cH)$ are {\em identically correlated}  
in $\rho\in\cS(\cH)$ if $\Tr[E^{X}(x)E^{Y}(y)\rh]=0$ for any $x,y\in\R$ with $x\ne y$. We 
say that two observables $X,Y\in\cO(\cH)$ are {\em identically distributed}  in a state 
$\rh\in\cS(\cH)$ if $\Tr[E^{X}(x)\rho]=\Tr[E^{Y}(x)\rho]$ for any $x\in\R$. The {\em cyclic subspace} of 
$\cH$ spanned by an observable $X$ and a state $\rho$ is the linear subspace $\cC(X,\rho)$ 
defined by
\beqa
\cC(X,\rho)
=\{p(X)\ps\mid \mb{$p(X)$ is a polynomial in $X$ and 
$\ps\in\ran(\rho)$}\}^{\perp\perp},
\eeqa
or equivalently $\cC(X,\rho)$ is the linear subspace spanned by
$X^{n}\ps$ for all $n=0,1,\ldots$ and all $\ps\in\ran(\rho)$.
Then we obtain the following theorem.

\begin{theorem}\label{th:ID'}
For any two observables $X,Y\in\cO(\cH)$ and any state $\rho\in\cS(\cH)$,
the following conditions are equivalent.

(i) $X$ and $Y$ have the same value in $\rho$.

(ii) $X$ and $Y$ are identically correlated in $\rho$.

(iii) $X$ and $Y$ are identically distributed in all $\psi\in\cC(X,\rho)$.

(iv) $f(X)\rho=f(Y)\rho$ for any function $f$.

(v) $X=Y$ on $\cC(X,\rho)$.
\end{theorem}
\begin{proof}
The equivalence between (i) and (ii) follows from Theorem \ref{th:ID} above
and Theorem  5.3 in Ref.~\cite{06QPC}.
The rest of the assertions follow from
Theorems 3.1, 3.2, and 3.4 in Ref.~\cite{06QPC}.
\end{proof}

The following theorem follows from Theorem \ref{th:ID'} and 
Theorem 4.4 in Ref.~\cite{06QPC}.

\begin{theorem} The relation $=_\rh$ is an equivalence 
relation on $\cO(\cH)$.
In particular, it is transitive, i.e., 
if $X=_\rh Y$ and $Y=_\rh Z$, then  $X=_\rh Z$
for all $X,Y,Z\in\cO(\cH)$.
\end{theorem}

The following theorem follows from Theorem \ref{th:ID'} and 
Theorems 5.8 and 6.3 in Ref.~\cite{07TPQ}.

\begin{theorem}
Let $X_1,\ldots,X_n\in\cO(\cH)$.  We have
\beqa
\val{X_1=X_2\And\cdots\And X_{n-1}=X_n}\le
\val{\cm(X_1,\ldots,X_n)}.
\eeqa
In particular, for any $\rh\in\cS(\cH)$, 
if we have $X_1=_{\rh}X_2, \ldots,  X_{n-1}=_{\rh}X_n$, then  
$X_1,\ldots,X_n$ are jointly determinate in $\rh$.
\end{theorem}

\section{Measuring processes}
\label{se:5}

A {\em measuring process} for $\cH$ is defined to be a quadruple
$(\cK,\si,U,M)$ consisting of a Hilbert space $\cK$, 
a state $\si$ on $\cK$, 
a unitary operator $U$ on $\cH\otimes\cK$,
and an observable $M$ on $\cK$ \cite{84QC}.
A measuring process $\M(\bx)=(\cK,\si,U,M)$ with {\em output variable} $\bx$
describes a
measurement carried out by an interaction, called the {\em measuring interaction}, 
between the measured system
$\bS=\bS(\cH)$ described by $\cH$ 
and the {\em probe} system $\bP=\bS(\cK)$ described by $\cK$ 
that is prepared in the state $\si$ just before the measuring interaction.
The unitary operator $U$ describes the time evolution during the measuring 
interaction. The outcome of this measurement is obtained by
measuring the observable $M$, called the {\em meter observable},  
in the probe at the time just
after the measuring interaction.  Thus, the {\em output distribution}
$\Pr\{\bx=x\|\rho\}$,
the probability distribution 
of the output variable $\bx$ of this measurement on input state $\rh\in\cS(\cH)$, 
is naturally defined by
\beqa
\Pr\{\bx=x\|\rho\}=\Pr\{U^{\da}(I\otimes M)U=x\|\rho\otimes\si\}.
\eeqa
Then, we have 
\beqa
\Pr\{\bx=x\|\rho\}&=&\Pr\{I\otimes M=x\|U(\rho\otimes\si)U^{\da}\}\nonumber\\
&=&\Tr[(I\otimes E^{M}(x))U(\rho\otimes\si)U^{\da}].
\eeqa
Moreover, the {\em quantum state reduction}, the state change from the input
state $\rh$ to
the state $\rho_{\{\bx=x\}}$ just after the measurement 
under the condition
$\bx=x$, is given by
\beqa
\rho\mapsto\rho_{\{\bx=x\}}
=
\frac{
\Tr_{\cK}[(I\otimes
E^{M}(x))U(\rho\otimes\si)U^{\da}]}
{\Tr[(I\otimes E^{M}(x))U(\rho\otimes\si)U^{\da}]},
\eeqa
where $\Tr_{\cK}$ stands for the partial trace over $\cK$,
provided that $\Pr\{\bx=x\|\rho\}>0$.

An {\em instrument} is a mapping $\cI$ from $\R$ to the space ${\rm CP}(\cH)$
of completely positive maps on the space $\cL(\cH)$;
for general theory of instruments see Ref.~\cite{84QC,DL70}. 
Let $\M(\bx)=(\cK,\si,U,M)$ be a measuring process with output variable $\bx$.
The {\em instrument} of $\M(\bx)$ is defined by
\beqa\label{eq:instrument}
\cI(x)\rh
=\Tr_{\cK}[(I\otimes E^{M}(x))U(\rh\otimes\si)U^{\da}].
\eeqa 
The POVM of  $\M(\bx)$ is defined by
\beqa\label{eq:POVM}
\Pi(x)
=\Tr_{\cK}[U^{\da}(I\otimes E^{M}(x))U(I\otimes\si)].
\eeqa 
They satisfy the relations \cite{84QC}
\beqa
\Pr\{\bx=x\|\rh\}&=&\Tr[\cI(x)\rho]=\Tr[\Pi(x)\rho],\\
\rho_{\{\bx=x\}}&=&
\frac{\cI(x)\rho}{\Tr[\cI(x)\rho]}.
\eeqa

The dual map $T^{*}$ of $T\in {\rm CP}(\cH)$ is defined by 
$\Tr[(T^{*}a)\rh]=\Tr[a(T\rh)]$.
The POVM $\Pi$ and the instrument $\cI$ are related by
 the relation \cite{DL70}
\beqa\label{eq:P-I}
\Pi(x)=\cI(x)^{*}1,
\eeqa
where $\cI(x)^{*}$ is the dual map of $\cI(x)$.
Conversely, it has been proved in Ref.~\cite{84QC} that
for every POVM $\Pi$ there exists an instrument $\cI$
satisfying \Eq{P-I}, and for every instrument $\cI$ there is a measuring
process $(\cK,\si,U,M)$ satisfying \Eq{instrument}.
For further information on the theory of quantum measurements based on
the notions of measuring processes, instruments, and POVMs, we refer 
the reader to Refs.~\cite{DL70,Dav76,84QC,04URN}.

\section{Measurements of observables}
\label{se:6}

A measuring process $\M(\bx)=(\cK,\si,U,M)$ 
with output variable $\bx$ is 
said to be a {\em measurement of an observable $A$ in a state $\rho$},
or said to {\em measure $A$ in $\rho$},
if $A\otimes I$ 
and $U^{\da}(I\otimes M)U$ has the same value in the state $\rho\otimes\si$.
In the measuring process $\M(\bx)$, 
the observer actually measures the meter observable 
just after the measuring interaction, 
but reports the value of this observable as the value of the observable 
$A$ just before the measuring interaction;
or in other words the observer actually measures $U^{\da}(I\otimes M)U$ in
the state $\rho\otimes\si$, but reports the outcome to be the value of 
the observable $A\otimes I$  in the state $\rho\otimes\si$.
Thus, this procedure is justified if and only if 
$A\otimes I$ and $U^{\da}(I\otimes M)U$ has the same value 
in the state $\rho\otimes\si$.

A measuring process $\M(\bx)=(\cK,\si,U,M)$ is 
said to satisfy the {\em Born statistical formula} (BSF)  for $A\in\cO(\cH)$
in $\rho\in\cS(\cH)$ if it satisfies
$
\Pr\{\bx=x\|\rho\}=\Tr[E^{A}(x)\rho]
$
for all $x\in\R$.
A POVM $\Pi$ is said to be {\em identically correlated} with
an observable $A$ in $\rho$ if
$
\Tr[\Pi(x)E^{A}(y)\rho]=0
$
for any $x,y\in\R$ with $x\ne y$.

The following theorem characterizes measurements 
of an observable in a given state.

\begin{theorem}
Let $\M(\bx)=(\cK,\si,U,M)$ be a measuring process for $\cH$.
Let $\Pi$ be the POVM of $\M(\bx)$.
For any observable $A$ and any state $\rh$,
the following conditions are all equivalent.

(i) $\M(\bx)$ is a measurement of $A$ in $\rho$.

(ii) $\M(\bx)$ satisfies the BSF for $A$ in any $\psi\in\cC(A,\rho)$.

(iii) $\Pi$ is identically correlated with $A$ in $\rho$.

\end{theorem}
\begin{proof}
The assertion follows from Theorem 8.2 of Ref.~\cite{06QPC}.
\end{proof}

\begin{theorem}
A measuring process $\M(\bx)=(\cK,\si,U,M)$ for $\cH$ is a measurement of 
an observable $A\in\cO(\cH)$ in any state if and only if its POVM
$\Pi$ coincides with the spectral measure of $A$, i.e., $\Pi=E^{A}$.
\end{theorem}
\begin{proof}
The assertion follows immediately from the fact that 
$\Pi$ is identically correlated with $A$ in any state
if and  only if $\Pi(x)=E^{A}(x)$ for all $x\in\R$.
\end{proof}

\section{Simultaneous measurability}
\label{se:7}

For any measuring process $\M(\bx)=(\cK,\si,U,M)$ for $\cH$ with output variable 
$\bx$ and a real function $f$, the measuring process $\M(f(\bx))$ with output variable 
$f(\bx)$ is defined by $\M(f(\bx))=(\cK,\si,U,f(M))$.  Observables $A_1,\ldots,A_n$ are 
said to be {\em simultaneously measurable} in a state $\rh\in\cS(\cH)$
by $\M(\bx)$ if there are real functions $f_1,\ldots,f_n$ such that
$\M(f_j(\bx))$ measures $A_j$ in $\rh$ for $j=1,\ldots,n$. 
Observables $A_1,\ldots,A_n$ are said to be {\em simultaneously measurable}
 in $\rh$ if there is a measuring process $\M(\bx)$ such that 
$A_1,\ldots,A_n$ are simultaneously measurable in $\rh$ by $\M(\bx)$.
 
The cyclic subspace $\cC(A,B,\rho)$ generated by $A,B,\rho$ is
defined by 
\beqa
\lefteqn{
\cC(A,B,\rho)}\quad\nonumber\\
&=&\{p(A,B)\ps\mid \mb{$p(A,B)$ is a polynomial in $A,B$
and $\ps\in\ran(\rho)$}\}^{\perp\perp}.
\eeqa
The simultaneous measurability and the commutativity are not
equivalent notion under the state-dependent formulation, as
the following theorem clarifies.

\begin{theorem}\label{th:main}
(i) Two observables $A,B\in\cO(\cH)$ are jointly determinate in a state $\rh\in\cS(\cH)$ if and only if
there is an $\R^{2}$-valued POVM $\Pi$ such that 
\beqa
\sum_y\Pi(x,y)&=&E^{A}(x)\quad\mbox{on}\quad\cC(A,B,\rho),
\label{eq:i-1}\\
\sum_x\Pi(x,y)&=&E^{B}(y)\quad\mbox{on}\quad\cC(A,B,\rho).
\label{eq:i-2}
\eeqa

(ii) Two observables $A,B\in\cO(\cH)$ are
simultaneously measurable in a state $\rh\in\cS(\cH)$
 if and only
if there is an $\R^{2}$-valued POVM $\Pi$ such that 
\beqa
\sum_y\Pi(x,y)&=&E^{A}(x)\quad\mbox{on}\quad\cC(A,\rho),
\label{eq:ii-1}\\
\sum_x\Pi(x,y)&=&E^{B}(y)\quad\mbox{on}\quad\cC(B,\rho).
\label{eq:ii-2}
\eeqa

(iii) Two observables on $\cH$ are
simultaneously measurable in a state $\rh\in\cS(\cH)$
if they are jointly determinate in $\rh$. 
\end{theorem}
\begin{proof}
Suppose that there is an $\R^{2}$-valued POVM $\Pi$ satisfying 
Eqs.~\eq{i-1} and \eq{i-2}.
Let $P=\cP(\cC(A,B,\rh))$. 
Let $\Pi'$ be such that $\Pi'(x,y)=P\Pi(x,y)P$.  
Then, marginals of $\Pi'$ are projection-valued measures
$E^{A}(x)P$ and $E^{B}(y)P$,
so that from a well-know theorem, e.g., Theorem 3.2.1 of Ref.~\cite{Dav76}, 
the marginals commute and $\Pi'$ is the product of the marginals.
Thus, we have
\beqas
P\Pi(x,y)P&=&\Pi'(x,y)=E^{A}(x)PE^{B}(y)P=
(E^{A}(x)P\And E^{B}(y)P)\\
&=&
(E^{A}(x)\And E^{B}(y))P
\eeqas
Let $\mu(x,y)=\Tr[\Pi(x,y)\rh]$.  
Since $\Pi$ is  a POVM, $\mu$ is a probability distribution.
Since $P\rh=\rh P=\rh$, we have 
$
\mu(x,y)=\Tr[(E^{A}(x)\And E^{B}(y))\rh],
$
so that $A$ and $B$ has the joint probability distribution in $\rh$
and they are jointly determinate.
Conversely, suppose that $A$ and $B$ are jointly determinate in a state $\rh$.
Then, $\ran{\rh}\subseteq\cR(\val{\cm(A,B)})$ and $\cR(\val{\cm(A,B)})$
is both $A$-invariant and $B$-invariant, so that we have 
$\cC(A,B,\rho)\subseteq\cR(\val{\cm(A,B)})$.
Then, we have $AB=BA$ on $\cC(A,B,\rho)$ and 
$\Pi'(x,y)=(E^{A}(x)\And E^{B}(y))P$ can be extended to a POVM $\Pi$
satisfying Eqs.~\eq{i-1} and \eq{i-2}.
Thus, statement (i) follows.
Suppose that $A,B\in\cO(\cH)$ are
simultaneously measurable in $\rh\in\cS(\cH)$.
Then, we have a measuring process $\M(\bx)=(\cK,\si,U,M)$ and real functions
$f,g$ such that $\M(f(\bx))$ measures $A$ in $\rh$ and 
$\M(g(\bx))$ measures $B$ in $\rh$.
Let $\Pi_0(x)$ be the POVM of  $\M(\bx)$.
Let $\Pi(x,y)=\sum_{x':(x,y)=(f(x'),g(x'))}\Pi_0(x')$.
Then, it is easy to see that $\Pi$ satisfies Eqs.~\eq{ii-1} and \eq{ii-2}.
Conversely, suppose that there is 
an $\R^{2}$-valued POVM $\Pi$ satisfying Eqs.~\eq{ii-1} and \eq{ii-2}.
Then, from the realization theorem of instruments and POVMs,
Theorem 5.1 of Ref.~\cite{84QC}, 
we have a measuring process 
$\M(\bx)=(\cK,\si,U,M)$ and real functions
$f,g$ such that 
\beqa
\Pi(x,y)=\Tr_{\cK}
[U^{\da}(I\otimes \sum_{x':(x,y)=(f(x'),g(x'))}E^{M}(x'))U(I\otimes \si)].
\eeqa
Then, we have 
\beqa
\sum_{y}\Pi(x,y)
=\Tr_{\cK}
[U^{\da}(I\otimes E^{f(M)}(x))U(I\otimes \si)],
\eeqa
so that from \Eq{ii-1} we have $\M(f(\bx))$ measures $A$ in $\rh$.
Similarly, we can show that $\M(g(\bx))$ measures $B$ in $\rh$.
Thus, statement (ii) follows.
Now, statement (iii) follows (i) and (ii)
\end{proof}

The conventional relation between simultaneous measurability and
commutativity, or joint determinateness,  
in the sate-independent formulation 
is recovered in our quantum logical approach as follows.

\begin{theorem}
For any observables $A,B\in\cO(\cH)$,
the following conditions are all equivalent.

(i) $A$ and $B$ are jointly determinate in any state $\rh\in\cS(\cH)$.

(ii) $A$ and $B$ are simultaneously measurable in any state $\rh\in\cS(\cH)$.

(iii) $A$ and $B$ commute on $\cH$.
\end{theorem}
\begin{proof}
From Theorem \ref{th:main} conditions (i) and (ii) are both equivalent to that
there is an $\R^{2}$-valued POVM $\Pi$ such that their marginals are
equal to $E^{A}$ and $E^{B}$.  This condition is well-known to be 
equivalent to condition (iii); see for example Theorem 3.2.1 of Ref.~\cite{Dav76}.
\end{proof}

An observable is said to be {\em non-degenerate} if every eigenvalue has 
one-dimensional eigenspace.  In the case where $\dim(\cH)=2$, every 
observable is non-degenerate or scalar.

\begin{theorem}\label{th:12}
Suppose that  $\dim(\cH)=2$.
For any non-degenerate observables $A,B\in\cO(\cH)$ and any state $\rh\in\cS(\cH)$
that is not an eigenstate of $A$ or $B$,
the following conditions are all equivalent.

(i) $A$ and $B$ are jointly determinate in $\rh$.

(ii) $A$ and $B$ are simultaneously measurable in $\rh$.

(iii) $A$ and $B$ commute on $\cH$.
\end{theorem}
\begin{proof}
In this case, we have $\cC(A,\rho)=\cC(B,\rh)=\cC(A,B,\rh)=\cH$,
and hence the assertion follows easily.
\end{proof}

The following theorems show that we can simultaneously measure 
two nowhere commuting observables.
   
\begin{theorem}\label{th:13}
In any Hilbert space,
every pair of observables are simultaneously measurable
in any eigenstate of either observable.
\end{theorem}
\begin{proof}
Let $A,B\in\cO(\cH)$.
Suppose that we have $A\ps=a\ps$ and $\|\ps\|=1$ with $\ps\in\cH$ and $a\in\R$.
Then, $\Pi(x,y)=\de_{x,a}\ketbra{\ps}E^{B}(y)$ satisfies Eqs.~\eq{ii-1} and \eq{ii-2},
so that $A$ and $B$ are simultaneously measurable in $\ps$.
\end{proof}

From Theorems \ref{th:12} and \ref{th:13}, we can characterize all pairs of 
simultaneously 
measurable observables in the case where $\dim(\cH)=2$, as follows.
If $A$ and $B$ commute, then they are simultaneously measurable in every state.
If $A$ and $B$ do not commute, then they are non-degenerate, and hence 
simultaneously measurable if and only if the state is an eigenstate of $A$ or $B$.  

\begin{theorem}
In any Hilbert space with dimension more than 3,
there are nowhere commuting observables
that are simultaneously measurable in a state that is
not an eigenstate of either observable.
\end{theorem}
\begin{proof}
Let $\cH$  be a Hilbert space with $\dim(\cH)=n>3$.
First we suppose $n=4$.
In this case we can assume without any loss of generality 
that $\cH$ is a Hilbert space of a pair of spin 1/2
particles, i.e., $\cH=\C^2\otimes \C^2$.
Obviously, the observables $A=\si_z\otimes I$ and $B=\si_x\otimes I$ 
on $\cH$ are nowhere commuting.
Let 
\beqa\label{eq:EPR-3}
\ps=\frac{1}{\sqrt{2}}
(\ket{\si_x=+1}\ket{\si_x=+1}+\ket{\si_x=-1}\ket{\si_x=-1}).
\eeqa
Let  $C=I\otimes \si_x$.  
Let $\Pi(x,y)=E^{\si_z}(x)\otimes E^{\si_x}(y)$.
Then $\Pi(x,y)$ is  a  $\R^2$-valued POVM satisfying
\beqa
\sum_{y}\Pi(x,y)&=&E^{A}(x),
\label{eq:EPR-1}\\
\sum_{x}\Pi(x,y)&=&E^{C}(y).
\label{eq:EPR-2}
\eeqa
From \Eq{EPR-3} we have
$
C=_{\ps}B,
$
so that $E^{C}(y)=E^{B}(y)$ on $\cC(B,\ps)$.
By the transitivity of the relation $=_{\psi}$, we have
\beqa
\sum_{x}\Pi(x,y)&=&E^{B}(y)
\eeqa
on $\cC(B,\ps)$.
Thus, $A$ and $B$ are simultaneously measurable.
In the general case, we can assume that the space
$\C^2\otimes \C^2$ is a subspace of $\cH$.
Then, it is easy to see that $A$ and $B$ are extended to
two nowhere commuting observables on $\cH$
and that they are simultaneously measurable in the state
$\psi$.
\end{proof}

\section{Discussions and conclusions}

We have considered the language for propositional logic constructed from
atomic formulas of the form $A=a$, where $A$ denotes an observable of
a quantum system described by a Hilbert space $\cH$ and $a$ denotes a real 
number.  
Every sentence in this language is called an observational proposition for $\cH$.
To every observational proposition $\ph$ we have assigned the 
$\cQ(\cH)$-valued truth value $\val{\ph}$, 
where $\cQ(\cH)$ is the lattice of projections on $\cH$. 
It is easy to check that this assignment of observational propositions to 
projections
on $\cH$, or equivalently to subspaces of $\cH$, is equivalent to the 
assignment
first proposed by Birkhoff and von Neumann \cite{BN36}.
Then, for any state $\rh$ we have assigned 
to every observational proposition $\ph$
the probability $\Pr\{\ph\|\rh\}$,
which is consistent with the Born statistical formula. 
We say that an observational proposition $\ph$ holds in a state $\rh$ 
if $\Pr\{\ph\|\rh\}=1$.
In the case where $\rh$ is a vector state
$\rh=\ketbra{\ps}$, 
$\Pr\{\ph\|\rh\}$ is the squared length of the vector $\val{\ph}\ps$,
and 
the observational proposition
$\ph$ holds in the state $\rh$ if and only $\ps\in\cR(\val{\ph})$.

Our language of observational propositions obviously commits the notion of
the values of observables.
Our language gives a systematic way to determine which 
sentences on the values of observables are true, partially true, or false
based on quantum logic.
The notion of the values of observables has been known to involve 
a difficulty suggested by the Kochen-Specker theorem \cite{KS67},
if we treat this notion in classical logic.
A key notion in our language that circumvents this difficulty
is the notion of joint determinateness.  This notion determines
the states in which given observables $A_1,\ldots,A_n$ have simultaneous
values as follows. We have introduced an observational proposition
$\cm(A_1,\ldots,A_n)$ defined by \Eq{SD}, and we say that
$A_1,\ldots,A_n$ are jointly determinate in state $\rh$
if $\cm(A_1,\ldots,A_n)$ holds in $\rh$.  This condition 
for a vector state $\rh=\ketbra{\ps}$ is equivalent
to that $\ps$ is a superposition of common eigenvectors of 
$A_1,\ldots,A_n$.  Since this proposition is a built-in proposition
in the sense that it is constructed from atomic formulas, 
quantum logic can determine the limitation for the notion of 
values of observables in its own right.

If our language could be interpreted in classical logic, 
we would have sufficiently many two-valued 
truth-value
assignments for all the sentences in the language,
and each assignment would give 
any observable $A$ only one real number $a$ such that
$A=a$ is true, so that $A$ is assigned the value $a$;
and then in any state the probability of the 
observational proposition would be interpreted 
as the ignorance as to which truth-value is 
actually assigned.

What the Kochen-Specker theorem denies is the existence
of such a two-valued assignment.
Instead, quantum logic allows contextual two-valued
assignments in the sense that if $A_1,\ldots,A_n$ are
jointly determinate in a state $\rh$, then we have
sufficiently many two-valued assignment
for all sentences in the sublanguage constructed by
atomic formulas of the form $A_1=a_1$, $\ldots$, or $A=a_n$,
where $a_1,\ldots,a_n$ denote arbitrary real numbers;
and then the probability of an observational proposition in this 
sublanguage in the state $\rh$ allows ignorance 
interpretation.  
The observational propositions in this sublanguage are called 
contextually well-formed for $A_1,\ldots,A_n$ in $\rh$.
Thus, in quantum logic the probability 
assignment is non-contextual but the two-valued
assignment is contextual.
This corresponds to the fact that the two-valued
assignment commits a result of a single measurement,
but the probability assignment commits only the statistics of 
results of many measurements.

One of our purpose of developing the language of observational propositions
is to extract another built-in proposition $A=B$ defined by \Eq{ID}
meaning that the observable $A$ and the observable $B$ has the same 
value.  
We say that $A$ and $B$ has the same value in state $\rh$ if
$A=B$ holds in $\rh$.
This condition for a vector state $\rh=\ketbra{\ps}$ is equivalent
to that $\ps$ is a superposition of common eigenvectors
of $A$ and $B$ belonging to the same eigenvalues;
in general $\rh$ is a mixture of such vector states.
As naturally expected, if $A=B$ holds in $\rh$ then $A$ and $B$ are
jointly determinate in $\rh$.
It is easy to see that if $A=B$ holds in $\rh$, 
we can simultaneously measure $A$ and $B$ in $\rh$, 
and each measurement gives the same measured value.
Moreover, in this case, we can simultaneously measure $A$ and 
any polynomial $f(A,B)$, so that if  we obtain the measurement
outcome $A=a$ then we obtain the measurement outcome
$f(A,B)=f(a,a)$.

Now, we are in a position to discuss the relation between 
quantum reality and measurement.
Every measurement is statistically equivalent to a model
describing a physical interaction
between the measured object and the probe followed by a subsequent 
measurement of the meter observable in the probe.
In this model, the measurement of the observable $A(0)$ 
in the state $\rh$
is replaced by the measurement of the meter observable
$M(\De t)$ after the interaction.  

In the conventional approach, the measurement of $A(0)$ is
considered to be correct if and only if  $A(0)$ and $M(\De t)$ have
the same probability distribution in the initial state $\rh\otimes\si$
for an arbitrary $\rh$ and a fixed $\si$, or equivalently the POVM
of the measurement coincides with the spectral measure of the 
measured observable. 
However, then the question arises about the status of the measured
value: how the measured value commits the reality of the measured system
just before the measurement.
A natural requirement is that the measured value
obtained from the $M(\De t)$ measurement reproduce the 
value of $A(0)$, but this requirement has a difficulty 
due to the Kochen-Specker theorem, which prohibits a
context-free assignment of the values of observables.
Our approach circumvents this difficulty using quantum logic.
In order to ensure that the given measuring process reproduces 
the value of measured observable before the interaction, we do not need
to assume the context-fee assignment.
Instead, we can justify it by our well-founded statement 
that  $A(0)$ and $M(\De t)$ have the same value in the state  $\rh\otimes\si$.
In that, we do not need a context-free value assignment, but
it is only required that there should be a context in which the values
of $A(0)$ and $M(\De t)$ are jointly assigned and they are the same.
We call this value the outcome obtained by the measurement.
Actually, this approach justifies the conventional definition of 
correct measurements of an observable as follows:
the measuring process reproduces the probability distribution in any 
state $\rh$ if and only if the measuring process also reproduces 
the value of the measured observable in any state $\rh$.
Thus, we can speak of the value of observables in a context-free
language based on quantum logic, and this language consistently
implies the reality of the measurement outcome in the contextual
language based on classical logic.

A major significance of our approach is that it gives us a criterion,
which does not exist in the conventional approach,
on what is the correct measurement of an observable in a given state.
As mentioned above, in the conventional approach the correct measurement
is characterized by the probability reproducibility in arbitrary states, 
but the probability reproducibility cannot be used as the criterion on
the correct measurement in a given state, as is obvious even in the 
classical case.
However, our new criterion applies to the measurement in a given sate
stating that the measurement of $A(0)$ 
in a given state $\rh$ is correct if and only if  $A(0)$ and $M(\De t)$ 
have the same value in the initial state $\rh\otimes\si$.

Another feature of our approach is to enable us to consider 
the state-dependent notion of simultaneous measurability.
The state-independent notion of simultaneous measurability is known
to be equivalent with the commutativity.  However, the state-dependent
notion of simultaneous measurability has not been given a right 
place in quantum mechanics, although a few has been considered 
as pathological exceptions of the uncertainty principle.

It has been claimed for long that if at time 0 the object is 
prepared in an
eigenstate of $A(0)$ and the observer actually measures the value of another 
observable $B(0)$ at time 0, then the observer can know both the 
values of two observables $A(0)$ and $B(0)$, even though
they are nowhere commuting.
Heisenberg discussed this case in his book \cite{Hei30}
published in 1930.
His reluctance to accept this case as a simultaneous measurement
is mainly due to the fact that this does not leave the system
in a joint eigenstate of $A(0)$ and $B(0)$.
However, the notions of state preparation and measurement
should have been clearly distinguished. In fact, it is widely accepted 
nowadays that any observable can be measured correctly without leaving 
the object in an eigenstate of the measured observable; for instance,
a projection $E$ can be correctly measured in a state $\ps$
with the outcome being 1 
leaving the object in the state $M\ps/\|M\ps\|$, where the
operator $M$ depends on the apparatus and satisfies $E=M^{\da}M$
(see, for example, a widely accepted text book by Nielsen and Chuang
\cite{NC00}).

On the other hand,
one of the variations of the Eeinstein-Podolsky-Rosen (EPR) argument
\cite{EPR35}  runs as follows.  
In the EPR state of two particles, I and II, the momentum
of particle I can be measured by directly and locally measuring 
the momentum of particle II taking into account the EPR correlation;
this follows from the EPR original argument stating that 
the locality of measurement ensures that the predicted correlation
determines the value of momentum of particle I. 
The locality of the momentum measurement of particle II 
also concludes that it does not disturb the particle I,
and hence we can simultaneously measure the position of particle I
by a direct measurement on particle I.
Thus, the momentum and position of particle I are simultaneously 
measurable, so that both the measured values corresponds to 
elements of reality.  However, quantum mechanics has no state 
to describe those results, and hence it should be incomplete.

In the conventional approach, we cannot discuss those cases 
in the light of general theory of quantum measurement, since
those simultaneous measurements are state-dependent.
In our approach, we have provided a general theory of 
state-dependent simultaneous measurements, and actually 
the above two cases are two special cases of simultaneous 
measurements characterized by our rigorous definition.

According to our theory, the simultaneous measurement of
$A(0)$ and $B(0)$ by the meter $M(\De t)$ is 
defined by the following two conditions:

(i)  $A(0)=f(M(\De t))$ holds in 
the state $\rh\otimes\si$.

(ii) $B(0)=g(M(\De t))$ holds in the state $\rh\otimes\si$.

From (i) we can conclude that $A(0)$ and $f(M(\De t))$ are
jointly determinate in $\rh\otimes\si$.
From (ii) the same is true for $B(0)$ and $g(M(\De t))$.
However, this does not imply that $A(0)$ and $B(0)$ are
jointly determinate in $\rh\otimes\si$.
Thus, the simultaneous measurability of $A(0)$ and $B(0)$
does not ensure that the two outcomes from the simultaneous
measurements has simultaneous reality.
This is because of the contextuality of the two defining 
conditions of simultaneous measurement.
Condition (i) ensures $A(0)=f(M(\De t))$ is contextually 
well-formed in $\rh\otimes\si$, and
condition (ii) ensures the same is true for $B(0)=g(M(\De t))$.
The statement ``$A(0)=f(M(\De t))$ and $B(0)=g(M(\De t))$''
holds in $\rh\otimes\si$, but this is not contextually well-formed
unless $A(0)$ and $B(0)$ are jointly determinate in 
$\rh\otimes\si$.
Thus, if the apparatus has made a simultaneous measurements
of nowhere commuting observables $A(0)$ and $B(0)$ 
and obtained the outcome $A(0)=a$ and $B(0)=b$,
we can use the fact $A(0)=a$ in one of the context 
including $A(0)$ as the reality of the measured object
subject to classical logic and the same is true for $B(0)=b$,
but we have no right to use both $A(0)=a$ and $B(0)=b$
as elements of the unified reality.

In conclusion, quantum logic sheds  a unique light on two 
facets of quantum mechanics.
The quantum state is defined as a probability 
measure on the lattice of observational propositions.
Then, the projection-valued truth-value assignment and 
the probability assignment, consistent with
the Born probabilistic interpretation, are non-contextual,
and explain our experience about the statistics of  results 
of many measurements.
This aspect of quantum mechanics has been
emphasized as the statistical interpretation typically 
formulated by Ballentine \cite{Bal70}.
On the other hand, the two-valued assignment to the
observational propositions is contextual,
and explains our experience about the results of single
measurements together with theoretical predictions 
in a context of sublanguage of observational propositions 
subject to classical logic.
This aspect of quantum mechanics has explained 
how we should apply our intuition about physical reality 
stemmed from classical physics to quantum mechanical objects,
and was typically stressed by Bohr's complementarity principle
\cite{Boh28}.

\section*{Acknowledgements}
This work was supported in part 
by the Grant-in-Aid for Scientific Research, No.~22654013
 and No.~21244007, of the JSPS.


\end{document}